\documentclass[twocolumn,10pt]{IEEEtran}
\usepackage{braket}

\input{def.tex}
\usepackage{dsfont}
\usepackage{minibox}
\DeclareSymbolFont{matha}{OML}{Txmi}{m}{it}
\DeclareMathSymbol{\varv}{\mathord}{matha}{118}
\usepackage{eurosym}
\usepackage{hhline,physics}

\usepackage{multicol}
\usepackage{algorithm,algorithmicx}
\usepackage{graphicx}
\usepackage{textcomp}
\usepackage{caption,pifont}
\usepackage[multiple]{footmisc}
\renewcommand{\baselinestretch}{0.85}

\usepackage[english]{babel}

\IEEEoverridecommandlockouts
\begin{document}
	\title{Coverage Probability of Double-IRS Assisted Communication Systems}
	\author{Anastasios Papazafeiropoulos, Pandelis Kourtessis, Symeon Chatzinotas, John M. Senior \thanks{A. Papazafeiropoulos is with the Communications and Intelligent Systems Research Group, University of Hertfordshire, Hatfield AL10 9AB, U. K., and with the  SnT at the University of Luxembourg, Luxembourg.  P. Kourtessis and John M.Senior are with the Communications and Intelligent Systems Research Group, University of Hertfordshire, Hatfield AL10 9AB, U. K. S. Chatzinotas is with the SnT at the University of Luxembourg, Luxembourg. This work was supported  by the University of Hertfordshire's 5-year Vice Chancellor's Research Fellowship and by the National Research Fund, Luxembourg, under the project RISOTTI. E-mails: tapapazaf@gmail.com, \{p.kourtessis,j.m.senior\}@herts.ac.uk, symeon.chatzinotas@uni.lu.}}
	\maketitle\vspace{-1.7cm}
	\begin{abstract}
		In this paper, we focus on the coverage probability of a double-intelligent reflecting surface (IRS) assisted wireless network and study the impact of multiplicative beamforming gain and correlated Rayleigh fading. In particular, we obtain a novel closed-form expression of the coverage probability of a single-input single-output (SISO) system assisted by two large IRSs while being dependent on the corresponding arbitrary reflecting beamforming matrices (RBMs) and large-scale statistics in terms of correlation matrices. Taking advantage of   the large-scale statistics, i.e., statistical channel state information (CSI), we perform optimization of the RBMs of both IRSs once per several coherence intervals rather than at each interval. Hence,		we achieve a reduction of the computational complexity, otherwise increased in multi-IRS-assisted networks during their RBM optimization. 	 Numerical results validate the analytical expressions even for small IRSs, confirm enhanced performance over the conventional single-IRS counterpart,  and reveal insightful properties.
	\end{abstract}
	
	\section{Introduction}
	Intelligent reflecting surface (IRS) is a recent hardware technology, which promises a significant increase in energy efficiency while achieving high spectral efficiency gains.  Its main advantage stems from its structure consisting of a large number of nearly passive elements that can shape the propagation environment since they can be digitally controlled by adaptively inducing amplitude changes and phase shifts on the impinging waves. Subsequently, IRSs have attracted a lot of research interest \cite{Wu2020,Bjoernson2019b,Kammoun2020, Elbir2020,Papazafeiropoulos2021c}.
	
	However, the majority of existing works have considered one or more independently distributed IRSs subject just to a single-signal reflection, while not accounting for any cooperation in terms of joint passive beamforming gain among them to  avoid the undesired interference. Recently, interesting works have  considered a double-IRS-assisted architecture and have studied the cooperative gain \cite{Han2020,Zheng2021,Zheng2021a}. In particular, in \cite{Han2020}, the cooperative beamforming gain in double-IRS assisted systems with a single-antenna transmitter and receiver is shown to be higher, i.e., $ \mathcal{O}(M^{4}) $ instead of  $ \mathcal{O}(M^{2}) $ appearing in single IRS systems \cite{Wu2019}, where $ M $ is the  total number of IRS elements. Also,  \cite{Han2020} relied on the assumption of only line-of-sight (LoS) existence for all links, while solely the double-reflection link between the two IRSs was considered with the single-reflection links being ignored. 	In this direction, by assuming a multi-user multi-antenna transmitter and that both single- and double-reflection links concur, authors in \cite{Zheng2021} and \cite{Zheng2021a} maximized the minimum signal-to-interference-plus-noise ratio and proposed an efficient channel estimation scheme, respectively.
	
	Although the study of the outage/coverage probability in IRS-assisted systems has attracted	a lot of attention \cite{Guo2020a,Yang2020,VanChien2021,Papazafeiropoulos2021a,Papazafeiropoulos2021b}, no work has studied this performance metric in double-IRS assisted systems. In parallel, previous works on double-IRS assisted systems were based on the independent Rayleigh fading  \cite{Han2020,Zheng2021}. However, this assumption is unrealistic in practice as shown in \cite{Bjoernson2020}.

	In comparison to previous literature on double-IRS assisted systems, we present for the first time a study of the coverage probability. To be specific,   we derive a closed-form expression of the signal-to-noise ratio (SNR) of a single-input single-output (SISO) system assisted by large IRSs by deterministic equivalent (DE) analysis. Next, we obtain a novel closed-form expression of the coverage probability. In particular, contrary to both \cite{Han2020, Zheng2021}, we have accounted for correlated Rayleigh fading instead of just LoS channels, and additionally to \cite{Han2020}, we have encountered single-reflection links, which is a more realistic assumption to identify the realistic potentials of the proposed architecture. Furthermore, our optimization approach, based on statistical CSI, is quite advantageous in terms of computational complexity compared to \cite{Zheng2021}, which relies on instantaneous CSI since it can take place once per several coherence intervals instead of at each coherence interval as required in instantaneous-CSI works.
	
	\textit{Notation}: Vectors and matrices are denoted by boldface lower and upper case symbols, respectively. The notations $(\cdot)^\T$, $(\cdot)^\H$, and $\tr\!\left( {\cdot} \right)$ represent the transpose, Hermitian transpose, and trace operators, respectively. Also, the notations $ \arg\left(\cdot\right) $ and $ \mod(\cdot,\cdot) $ denote the argument function and the modulus operation while the expectation operator is denoted by $\EE\left[\cdot\right]$. The notation  $\diag\left(\bA\right) $ expresses a vector with elements the diagonal elements of $ \bA $, while   $\mathrm{vec}(\bA) $ denotes a vector formed by stacking all the columns of $ \bA $ into
	a column vector. Given two infinite sequences $a_n$ and $b_n$, the relation $a_n\asymp b_n$ is equivalent to $a_n - b_n \xrightarrow[ n \rightarrow \infty]{\mbox{a.s.}} 0$. Finally, $\bb \sim \cC\cN{(\b0,\mathbf{\Sigma})}$ represents a circularly symmetric complex Gaussian vector with {zero mean} and covariance matrix $\mathbf{\Sigma}$.

	\section{System Model}\label{System}
	We consider a double IRS, cooperatively assisted, communication system between a single-antenna transmitter (Tx)  and a single-antenna receiver (Rx), where two spatially-distanced two-dimensional rectangular grid IRSs are placed in the locations of obstacles, found in the intermediate space, as shown in Fig. \ref{Fig0}. One IRS is closer to the Tx and the other closer to the Rx, and we refer to them as IRS $ 1 $ and IRS $2 $, respectively. This is a reasonable design to  exploit the IRS properties during their placement to enhance the communication, which is considered in double-IRS assisted works \cite{Han2020,Zheng2021,Zheng2021a}. Also, let a total of $ N $ reflecting elements in both IRSs with $ N_{i} $ denoting the number of elements of IRS $ i $, i.e, $ N_{1}+N_{2}=N $. 	 	In addition, we address the scenario with a blocked direct signal, where the Rx is served only through  reflection links by the IRSs. Each IRS is connected to a smart controller with the Tx through a separate perfect backhaul link to adjust its phases. 	To focus on fundamental properties of the double-IRS-assisted system, we rely on the assumption of perfect CSI, allowing more mathematical manipulations. Hence, the results play the role of upper bounds of practical implementations. In practice, the   CSI could be assumed perfectly known when the coherence intervals are sufficiently long. The consideration of  the imperfect CSI scenario, which is more practical, is the topic of future work.
	
	\begin{figure}[!h]
		\begin{center}
			\includegraphics[width=0.85\linewidth]{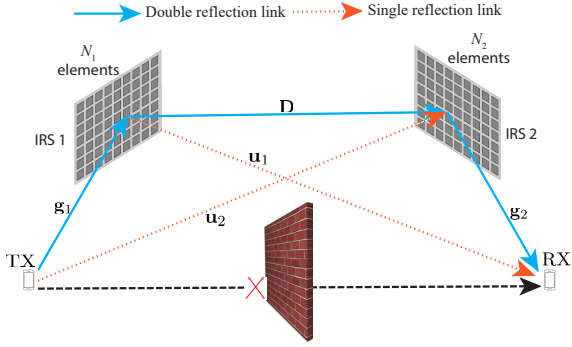}
			\caption{\footnotesize{A double-IRS cooperatively assisted SISO model. }}
			\label{Fig0}
		\end{center}
	\end{figure} 			
	Based on a  block-fading model with independent channel realizations across different coherence blocks, we denote $ 	\bg_{1}\in \mathbb{C}^{N_{1} \times 1}$,   $ 	\bD\in \mathbb{C}^{N_{2} \times N_{1}}$, $ 	\bg_{2}\in \mathbb{C}^{N_{2} \times 1}$, $ 	\bu_{1}\in \mathbb{C}^{N_{1} \times 1}$, and $ 	\bu_{2}\in \mathbb{C}^{N_{2} \times 1}$ as the channel vectors from the Tx to IRS $ 1 $, from IRS $ 1 $ to IRS $  2$, from IRS $  2$ to the Rx, from IRS $  1$ to the Rx, and from the Tx  to IRS $ 2 $, respectively. Taking into account for correlated Rayleigh fading and path-loss,\footnote{The analysis could be easily extended to account for  correlated Rician fading, where a  LoS component exists additionally to a multipath part but this  is the topic of future work.} we have $ \bg_{1}\sim \mathcal{CN}\left(\b0, \bR_{t1}\right) $,  $ \mathrm{vec}(\bD)\sim \mathcal{CN}\left(\b0,\bR_{12} \right) $, $ \bg_{2}\sim \mathcal{CN}\left(\b0, \bR_{2r}\right) $, $ 	\bu_{1}\sim \mathcal{CN}\left(\b0, \bR_{1r}\right) $, and $ 	\bu_{2}\sim \mathcal{CN}\left(\b0, \bR_{t2}\right) $, 	where 
	$\bR_{t1}=\beta_{t1} \bR_{1}\in \mathbb{C}^{N_{1} \times N_{1}} $, $ \bR_{12}=\beta_{12}\bR_{1}\otimes \bR_{2}\in \mathbb{C}^{N_{1}N_{2} \times N_{1}N_{2}}  $, $\bR_{2r}=\beta_{2r} \bR_{2}\in \mathbb{C}^{N_{2} \times N_{2}} $, $\bR_{1r}=\beta_{1r} \bR_{1}\in \mathbb{C}^{N_{1} \times N_{1}} $, and $\bR_{t2}=\beta_{t2} \bR_{2}\in \mathbb{C}^{N_{2} \times N_{2}} $ are the spatial covariance matrices of the respective links  with  $  \beta_{t1} $, $\beta_{12} $, $ \beta_{2r} $, $ \beta_{1r}$, and $  \beta_{t2}  $ being  the  path-losses and $ \bR_{1} \in \mathbb{C}^{N_{1} \times N_{1}}$, $ \bR_{2} \in \mathbb{C}^{N_{2} \times N_{2}}$ being the correlation matrices of IRS $ 1 $ and IRS $ 2 $. Note that the path-losses and the covariance matrices are assumed known while they can be obtained with practical methods, {e.g., see} \cite{Neumann2018}. Similar to \cite{Zheng2021,Zheng2021a}, we assume receive RF chains integrated
	into the IRSs. These chains have sensing abilities and enable the acquisition of the CSI of individual channels. 	In particular, concerning the correlation model, we have considered a suitable model for IRSs, which is obtained under the conditions of rectangular IRSs and isotropic Rayleigh fading \cite{Bjoernson2020}. Specifically, let  $d_\mathrm{V}$ and $d_\mathrm{H}$ denote the vertical height and  horizontal width of each IRS element (both IRSs are constructed in the same way and have the same size of elements). Then,  the $ \left(i,j\right) $th element of the corresponding correlation matrix $ \bR \in \mathbb{C}^{N \times N}  $ of the channels above in the case of an IRS with $ N=N_{\mathrm{H}}N_{\mathrm{V}}  $ elements is given by
	\begin{align} \label{eq:Element}
\bR_{ij} = d_{\mathrm{H}} d_{\mathrm{V}} \mathrm{sinc} \left( 2 \|\mathbf{u}_{i} - \mathbf{u}_{j} \|/\lambda\right),
	\end{align}
	where $\mathbf{u}_{\epsilon} = [0, \mod(\epsilon-1, N_\mathrm{H})d_\mathrm{H}, \lfloor (\epsilon-1)/N_\mathrm{H} \rfloor d_\mathrm{V}]^\T$, $\epsilon \in \{i,j\}$, and $\lambda$ is the wavelength of the plane wave while $ N_{\mathrm{H}} $ and  $ N_{\mathrm{V}} $ denote the elements per row   and per column of the two-dimensional rectangular IRS.
	 
	Relying on a slowly varying flat-fading narrowband channel model, the complex-valued received signal at the Rx through the double IRS assisted network  is expressed by
	\begin{align}
		y=\left(\bg_{1}^{\H}\bPhi_{1}\bD\bPhi_{2}\bg_{2}
		+\bg_{1}^{\H}\bPhi_{1}\bu_{1}+\bu_{2}^{\H}\bPhi_{2}\bg_{2}\right)x+z,\label{received}
	\end{align}
	where $ \bPhi_{i}=\mathrm{diag}\left(\al_{i1} \exp\left(j \theta_{i1}\right), \ldots, \al_{iN} \exp\left(j \theta_{iN_{i}}\right)\right)\in \mathbb{C}^{N_{i}\times N_{i}}$ expresses the reflecting beamforming matrix  of the $ i $th IRS, which depends on its elements with $ \theta_{in} \in \left[ 0, 2 \pi \right]$ and $ \al_{in} \in (0,1], n=1,\ldots,N_{i}$ being the phase shifts and the fixed amplitude reflection coefficients of the corresponding  element of the $ ith$ IRS ($ i=1,2 $). The progress on loss-less meta-surfaces allows to assume   maximum reflection, i.e., we set $ \al_{in}=1, \forall~i,n$ \cite{Bjoernson2019b}.\footnote{Recently, it was shown that the
		amplitude and phase responses are intertwined in practice, while the assumption of unity amplitude is unrealistic \cite{Abeywickrama2020}. However, the unity assumption amplitude still allows to reveal fundamental properties of the coverage probability of the proposed model.  The consideration of the phase shift model in \cite{Abeywickrama2020} suggests an interesting idea for extension of the 
		current work, i.e., to study the impact of coverage probability of double-IRS assisted systems by accounting for this intertwinement.} Moreover,  $ x $ expresses the transmitted data symbol satisfying $ \EE[|x|^{2}] \!=\!P$ with $ P $ being the average power of the symbol, and  $ z\sim\mathcal{CN}\left(0,N_{0}\right) $ expresses the additive white Gaussian noise (AWGN) sample.

	\section{Coverage Probability}\label{coverageProbability} 
	In this section, we present the derivation of the coverage probability when the double IRS assisted communication underlies to correlated Rayleigh fading 
	by applying DE analysis tools \cite{Papazafeiropoulos2015a}. 
	
	\subsection{Main Results}	
	The coverage probability, $ {P}_{\mathrm{c}} $, is defined as the probability, where the effective received SNR at the Rx, $ \gamma $, is larger than a given threshold $ T $. Mathematically described, we have  $ {P}_{\mathrm{c}}=\mathrm{Pr}\left(\gamma>T\right) $, where
	\begin{align}
		\gamma ={\gamma_{0}}{\big|\bg_{1}^{\H}\bPhi_{1}\bD\bPhi_{2}\bg_{2}
			+\bg_{1}^{\H}\bPhi_{1}\bu_{1}+\bu_{2}^{\H}\bPhi_{2}\bg_{2}\big|^{2}}{}\label{general}
	\end{align}
	is obtained based on \eqref{received} by assuming coherent communication and $ {\gamma_{0}}=P/N_{0} $ is the average transmit SNR. The exact derivation of the  probability function of this SNR is intractable while the  consideration of correlated Rayleigh fading hinders further the derivation. For this reason, we resort to the DE analysis to obtain the approximated DE SNR, which appears to match tightly  with $ \mathrm{Pr}\left(	\gamma>T\right) $ as shown in Section \ref{Numerical}. Notably, DE tools result in tight approximations even for finite practical dimensions, e.g., $ 8 \times 8 $ systems (see \cite{Papazafeiropoulos2015a} and references therein).  Obviously, this range is of practical interest in IRS-assisted systems. Also, the agreement of the analytical results with Monte-Carlo (MC) simulations in 		Section IV for finite $ N $ confirms this assertion. 	Hence, Proposition 1 below and the following results describe realistic
		systems of finite dimensions. Moreover, it is worthwhile to mention that even in the single-IRS-assisted scenario, the majority of works on IRS-assisted systems have provided approximations  based on the central limit theorem (CLT) that assumes large IRSs.  Hence, similar  to the existing relevant  literature \cite{Guo2020a,Yang2020,Papazafeiropoulos2021a},  we assume large IRSs ($ N_{i}\to \infty, $ $ i=1,2 $).
	
	\begin{proposition}\label{proposition:SNR}
		The DE SNR of a SISO transmission, enabled by a double IRS with correlated Rayleigh fading is given by
		\begin{align}
			\frac{1}{N_{1}N_{2}}{\gamma}
			&\asymp \bar{\gamma}\!,\!\label{DE_SNR1}
		\end{align}
		where 
		\begin{align}
				\bar{\gamma}\!&=\!\gamma_{0}\frac{1}{N_{1}N_{2}}\!\big(\!\tr\left(\bR_{t1}\bPhi_{1}\bR_{2}\bPhi_{1}^{\H} 	\right)\!\tr \!\left(\bR_{r2}\bPhi_{2}^{\H}\bR_{1}\bPhi_{2}\right)\nn\\
			&+\!\tr\!\left(\bR_{t1}\bPhi_{1}\bR_{1r}\bPhi_{1}^{\H}\right)\!+\!\tr\!\left(\bR_{t2}\bPhi_{2}\bR_{2r}\bPhi_{2}^{\H}\right)\!\!\big)\!.\!\label{DE_SNR} 
		\end{align}
	\end{proposition}
	\begin{proof}
		See Appendix~\ref{proposition1}.
	\end{proof}
	
	Obviously, the SNR in  \eqref{DE_SNR1} corresponds to a simple closed form that depends only on the RBMs of the corresponding IRSs and the statistical CSI in terms of the path-losses and the correlation matrices of the various channels involved. For example, the linear dependence from the path-losses shows how these degrade the SNR. Also, we observe that in the case of the unrealistic independent Rayleigh fading, where the correlation matrices are diagonal, the DE SNR does not depend on the phase-shifts, which means that it cannot be optimized.

	\begin{proposition}\label{proposition:Coverage}
		The coverage probability of a SISO transmission, enabled by a double IRS with correlated Rayleigh fading for arbitrary phase shifts, is tightly approximated as
		\begin{align}
			P_{\mathrm{c}}&\approx\sum^{M}_{n=1} \!\binom{M}{n}\!\left( -1 \right)^{n+1} \mathrm{e}^{ -n \eta \frac{T}{\bar{\gamma}}}\!\!\! \label{general1}		\\
			&=1-\big(1-\mathrm{e}^{-\frac{ \eta T}{	\bar{\gamma}_{}}}\big)^{\!M},\label{general2}	
		\end{align}	
		where   $\eta=M \left( M! \right)^{-\frac{1}{M}}$ with $ M $ being the number of terms used in the calculation.
	\end{proposition}
	\begin{proof}
		See Appendix~\ref{proposition2}.
	\end{proof}
	
	We observe that the coverage probability depends solely on the  threshold, DE SNR, and the number of terms used in the approximation. Note that the DE SNR  includes indirectly the RBMs, the path-losses, and the correlation matrices. Notably, in the case of independent Rayleigh fading, the coverage probability does not depend on the phase-shifts, and thus, the IRS is not exploitable to be optimized. We study their impact on $ P_{\mathrm{c}} $ in Sec. \ref{Numerical}.
	
	\subsection{Reflecting Beamforming Matrix Optimization}\label{3b}
	Herein, based on the common assumption of infinite resolution phase shifters, we propose an alternating optimization algorithm for designing the cooperative reflecting beamforming solving the problem of maximum $ P_{\mathrm{c}} $, which is formulated as
	\begin{align}\begin{split}
			&\!\!\!(\mathcal{P}1)~\max_{\bPhi_{1},\bPhi_{2}} ~~	P_{\mathrm{c}}\\
			&~~~~~~\mathrm{s.t}~|\phi_{in}|\!=\!1,~~ i\!=\!1,2~\mathrm{and}~n\!=\!1,\dots,N_{i},
		\end{split}\label{Maximization} 
	\end{align}
	where $ 	P_{\mathrm{c}} $ is given by \eqref{general1} or \eqref{general2} and $ \phi_{in}= \exp\left(j \theta_{in}\right) $. 
	
	Not only the problem $ (\mathcal{P}1) $ is non-convex and it is subject to unit-modulus constraints but also a multiplicative coupling between $ \bPhi_{1} $ and $ \bPhi_{2} $   appears in 	$ P_{\mathrm{c}} $ in terms of $ \bar{\gamma} $ given by \eqref{DE_SNR}. For this reason, we perform alternating optimization by optimizing one of the two RBMs  $ \bPhi_{i} $ while fixing the other in an iterative manner until reaching the convergence. Each optimization is achieved in terms of the  projected gradient 	ascent until converging to a 	stationary point. In particular, for given $ \bPhi_{1} $, at the $ l $th step, we consider the vector $ \bs_{1,l} =[\phi_{11}^{l}, \ldots, \phi_{1N_{1}}^{l}]^{\T}$, which contains 	the phases of IRS $ 1 $ at this step. In the next iteration,  $ P_{\mathrm{c}} $ increases until its convergence by projecting the solution onto the closest feasible point according to $ \min_{|\phi_{1n} |=1, n=1,\ldots,N_{1}}\|\bs_{1}-\tilde{\bs}_{1}\|^{2} $ while satisfying the unit-modulus constraint with respect to $ \phi_{1n} $. The next iteration point is described by
	\begin{align}
		\tilde{\bs}_{1,l+1}&=\bs_{1,l}+\mu \bq_{1,l},\label{sol1}\\
		\bs_{1,l+1}&=\exp\left(j \arg \left(\tilde{\bs}_{1,l+1}\right)\right),\label{sol2}
	\end{align}
	where $ \mu $ is the step size derived at each iteration by using the backtracking line search \cite{Boyd2004}. Also,  $ \bq_{1,l} $ expresses the  ascent direction at step $ l $, i.e., $ \bq_{1,l}= \pdv{	P_{\mathrm{c}}}{\bs_{1,l}^{*}} $, which is given by the following Lemma.	
	\begin{lemma}\label{deriv1}
		The derivative of the coverage probability with respect to $ \bs_{1,l}^{*}$ is given by
		\begin{align}
			\!\pdv{	{P_{\mathrm{c}}}}{\bs_{1,l}^{*}}\!\!&=\!\frac{\gamma_{0}}{N_{1}N_{2}}\!\sum^{M}_{n=1} \!\binom{M}{n}\! \frac{\left( -1 \right)^{n+1} n \eta T}{\bar{\gamma}^{2}\mathrm{e}^{ n \eta \frac{T}{\bar{\gamma}}}} \!\big(\diag\left(\bR_{t1}\bPhi_{1}\bR_{1r} \right)\nn\\
			&+\tr \!\left(\bR_{r2}\bPhi_{2}^{\H}\bR_{1}\bPhi_{2}\right)\diag\left(\bR_{t1}\bPhi_{1}\bR_{2} \right)\!\!\big).
		\end{align}
	\end{lemma}
	\begin{proof}
		See Appendix~\ref{ArbitraryPDFProof2}.
	\end{proof}

Lemma \ref{deriv1} shows that the corresponding  derivative of the coverage probability is obtained in closed-form, which depends mainly on the large-scale statistics, the number of elements of each IRS, and the variable $ M $ describing the number of approximation terms used during the derivation of $ P_{\mathrm{c}} $.

	 At each iteration step, $ P_{\mathrm{c}} $ increases until convergence to the optimum
	value since it is subject to a  power constraint. Next, we continue with the optimization with respect to $ \bPhi_{2} $. Note that since each IRS obeys to a similar solution, we focus on  IRS $ 1 $ while similar expressions hold for the optimization of IRS $ 2 $, which are omitted due to limited space. The RBM  design, based on the gradient ascent, results in an
	outstanding performance since the gradient ascent is obtained in a closed-from with low computational complexity due to the dependence on simple matrix operations. In particular, the complexity of (11) is $ \mathcal{O} \left(M+N^{3} \right) $. Obviously, 	it depends on  the parameters $ M $  and $ N $ but with a higher dependence on $ N $, especially, if $ N $ is large.
	
	The iteration of the overall process achieves
	the increase of $ P_{\mathrm{c}} $  until convergence of the objective to its optimum
	value since it is upper-bounded due to the power constraint. Given the non-convexity of the initial 	optimization problem with respect to the phase shifts, the proposed algorithm  does not
		provide a global optimum but a locally optimal solution, which is  a good preliminary tool to study the coverage probability of double-IRS aided systems.
	\subsection{Baseline Scenario: Conventional Single-IRS Assisted SISO system}
	We assume the baseline system, where all elements $ N=N_{1}+N_{2}  $ are allocated in a single IRS that could be found in the vicinity of the Tx or the Rx.  We denote $ \tilde{\bg}_{1}\sim \mathcal{CN}\big(\b0,\tilde{\bR}_{t}\big) $ and $ \tilde{\bg}_{2}\sim \mathcal{CN}\big(\b0, \tilde{\bR}_{r}\big) $ the baseband channels  of Tx-IRS and IRS-Rx links, respectively. Note that  $ \tilde{\beta}_{t} $, $ \tilde{\beta}_{r}$ and   $ \tilde{\bR}_{t}=\tilde{\beta}_{t}{\bR}\in \mathbb{C}^{N \times N} $, $ \tilde{\bR}_{r}=\tilde{\beta}_{r}{\bR} \in \mathbb{C}^{N \times N}$ are the path-losses and the correlation matrices, which are assumed to be known, while $ \bR=\bR_{1} $ 
since the the IRS of the baseline scenario is placed at the position of IRS $ 1 $. The corresponding RBM is described by $ \tilde{\bPhi}=\mathrm{diag}\big(\tilde{a}_{1} \exp(j \tilde{\theta}_{1}), \ldots,  \tilde{a}_{N}\exp(j \tilde{\theta}_{N})\big)\in \mathbb{C}^{N\times N}$, where $ \tilde{\theta}_{n} \in \left[ 0, 2 \pi \right], n=1,\ldots,N$ and $ \tilde{\al}_{n}=1~\forall n$ are the phase shifts and the fixed amplitude reflection coefficients of the $ n $th IRS element.
	
	The received SNR is written as
	\begin{align}
		\gamma^{*} =\gamma_{0}\big|\tilde{\bg}_{1}^{\H}\tilde{\bPhi}\tilde{\bg}_{2}\big|^{2}\label{general26}
	\end{align}
	while its DE expression, \cite[Lem. 4]{Papazafeiropoulos2015a}, is given by
	\begin{align}
		\frac{1}{N}\gamma^{*}
		&\asymp	\frac{1}{N}\gamma_{0}\tr\big(\tilde{\bR}_{t}\tilde{\bPhi}\tilde{\bR}_{r}\tilde{\bPhi}^{\H}\big)\label{general27}.
	\end{align}
	
	The coverage probability of the SISO channel with a single IRS $ P_{\mathrm{c}}^{*} $ is given by the  expression in \eqref{general1} after substituting $ \bar{\gamma} $ with $ \bar{\gamma}^{*} =\frac{1}{N}\gamma_{0}\tr\big(\tilde{\bR}_{t}\tilde{\bPhi}\tilde{\bR}_{r}\tilde{\bPhi}^{\H}\big)$. This expression can be optimized by applying again the gradient ascent method as in Subsection \ref{3b}. Specifically, to apply  the same algorithm for a single IRS-assisted SISO system, we require $ \pdv{	P_{\mathrm{c}}^{*}}{\bs_{l}^{*}} $, which is obtained as in Appendix \ref{ArbitraryPDFProof2} by 
	\begin{align}
	\!\!	\!\pdv{	{P_{\mathrm{c}}}}{\bs_{l}^{*}}\!\!&=\!\frac{\gamma_{0}}{N_{1}N_{2}}\!\sum^{M}_{n=1} \!\binom{M}{n}\! \frac{\left( -1 \right)^{n+1} n \eta T}{(\bar{\gamma}^{*})^{2}\mathrm{e}^{ n \eta \frac{T}{\bar{\gamma}^{*}}}} \!\diag\big(\tilde{\bR}_{t}\tilde{\bPhi}\tilde{\bR}_{r}\big).
	\end{align}
	
		\begin{figure*}[t]
		\begin{minipage}{0.33\textwidth}
			\centering
			\includegraphics[trim=0cm -0.20cm 0cm 0.2cm, clip=true, width=2.2in]{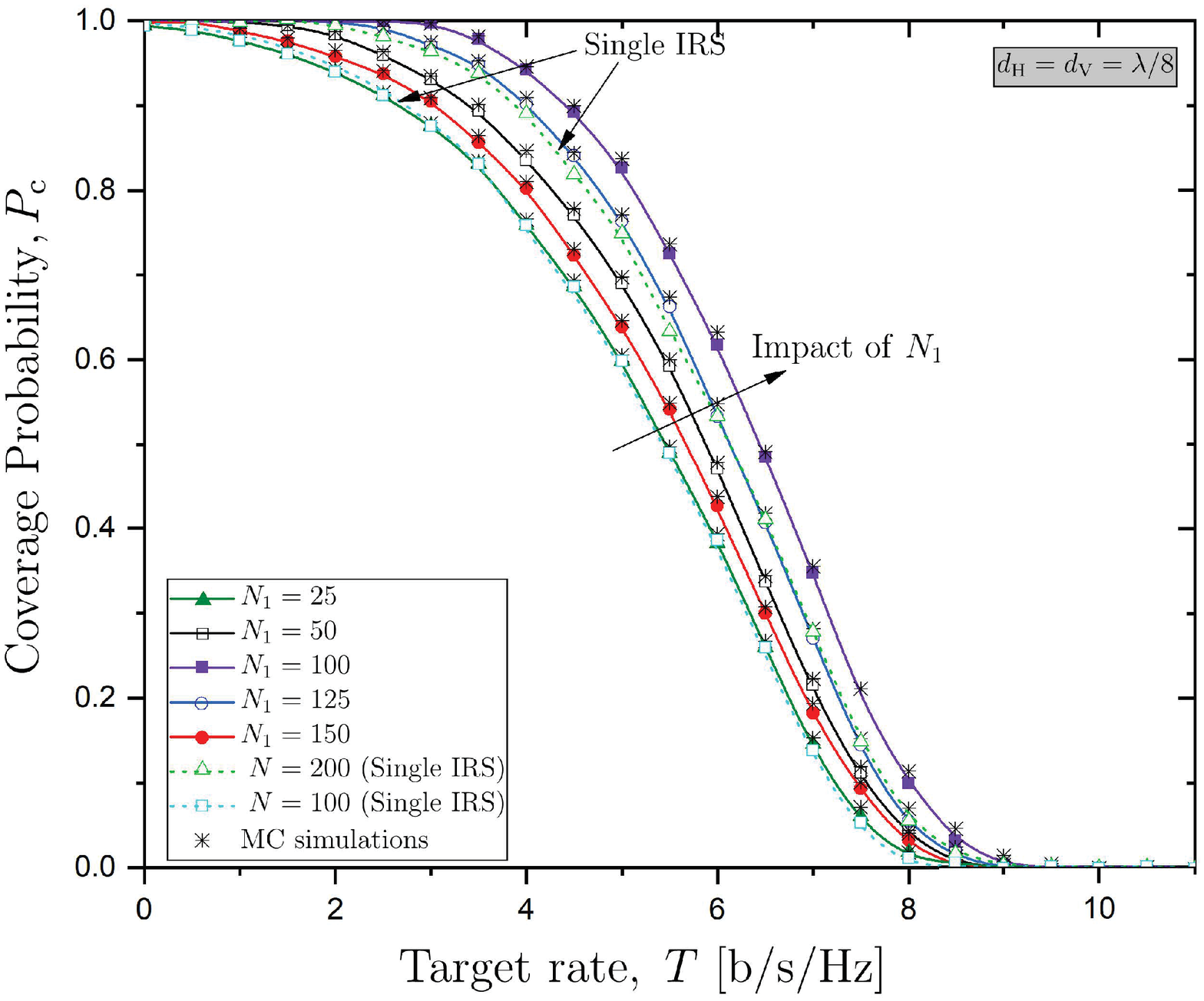} \vspace*{-0.1cm}
			\\ $(a)$
			\vspace*{-0.2cm}
		\end{minipage}
		\begin{minipage}{0.33\textwidth}
			\centering
			\includegraphics[trim=0cm -0.20cm 0cm 0.2cm, clip=true, width=2.2in]{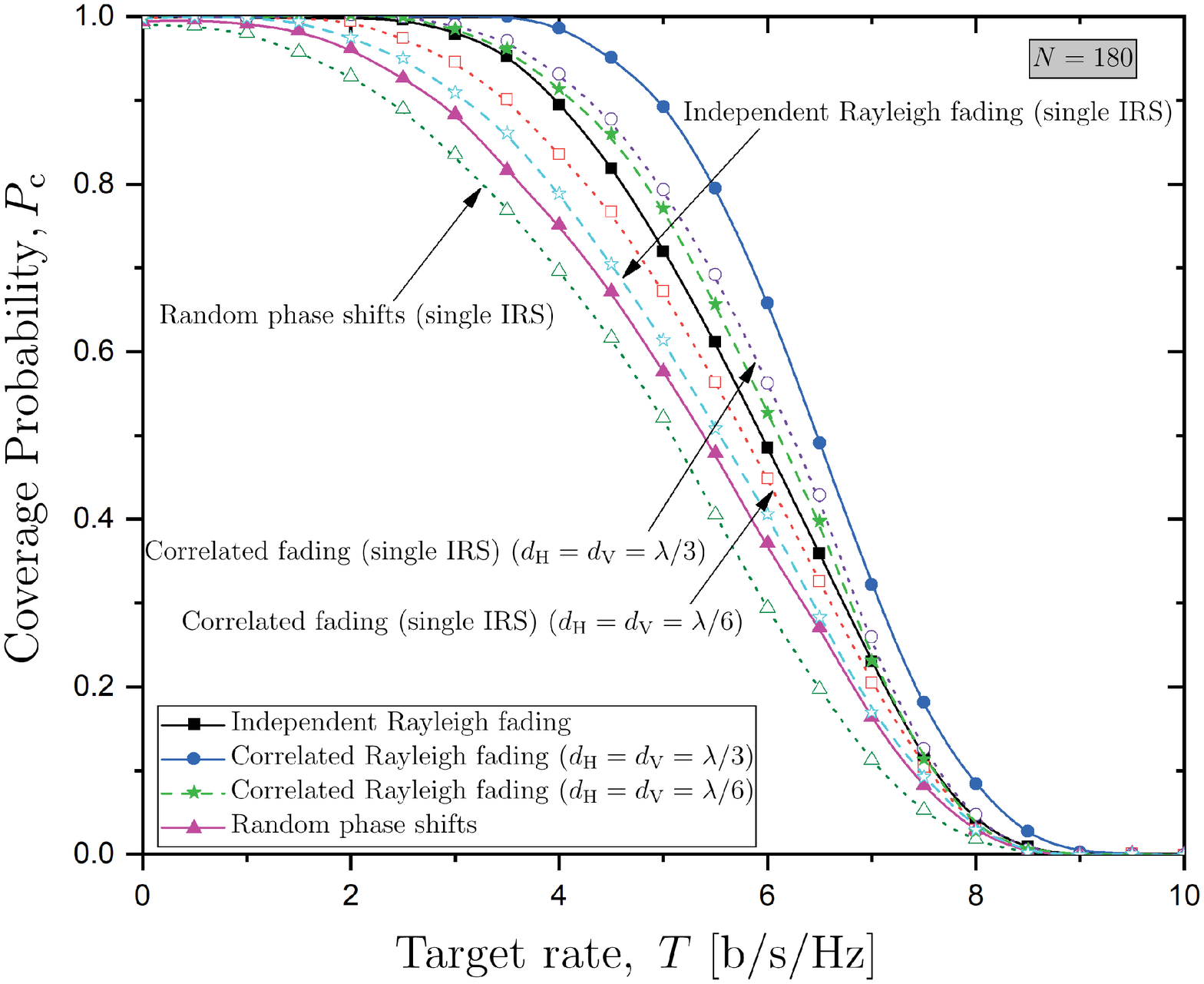} \vspace*{-0.1cm}
			\\$(b)$
			\vspace*{-0.2cm}
		\end{minipage}
		\begin{minipage}{0.33\textwidth}
			\centering
			\includegraphics[trim=0cm -0.20cm 0cm 0.2cm, clip=true, width=2.2in]{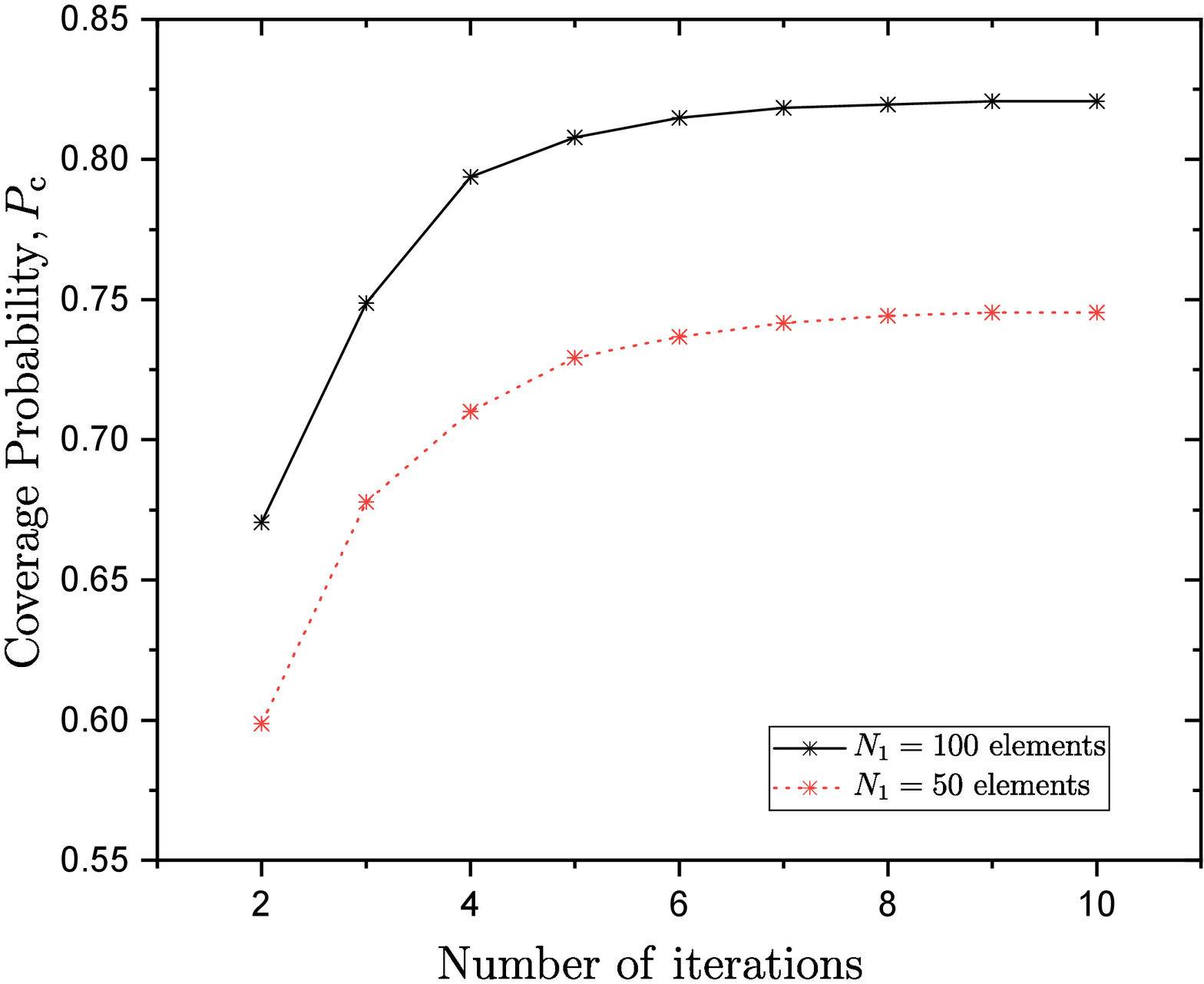}\vspace*{-0.1cm}\\$(c)$
			\vspace*{-0.2cm}
		\end{minipage}
		\caption{Coverage probability  of a double-IRS assisted SISO system with correlated Rayleigh fading versus the target rate $ T $ (analytical results and MC simulations)   $(a)$ for varying $ N_{1} $  ($ N=200 $, $ {d_{\mathrm{H}}=d_{\mathrm{V}}=\lambda/8} $ and $(b)$ for different correlation conditions and with/without optimization ($ N=180 $); $(c)$ Coverage probability versus the number of iterations for varying number of IRSs elements $N_{1}  $ ($ T=5\mathrm{~b/s/Hz} $).}
		\label{Fig1}
		\vspace{-0.65cm}
	\end{figure*}
	\section{Numerical Results}\label{Numerical}
	We consider a three-dimensional Cartesian coordinate system, where the locations of the Tx, Rx, IRS $ 1 $, and IRS $ 2 $ together with the directions of the IRSs follow the lines in \cite{Han2020} to avoid repetition and for the sake of comparison. Similarly, the distances between the Tx and IRS $ 1 $, between IRS $ 1 $ and IRS $ 2 $, and between IRS $ 2 $ and Rx are $ r_{t1}=1~\mathrm{m} $, $ r_{12}=100~\mathrm{m} $, and $ r_{2r}=15~\mathrm{m} $. In this work, similar to \cite{Han2020,Zheng2021}, we assume that the  IRS  in the single-IRS (baseline) case is located at a distance $ \tilde{r}_{t} \approx r_{12}$ to result in nearly the same “product distance” for the two cases, i.e., $  \tilde{r}_{t}  \tilde{r}_{r}\approx r_{t1}r_{12}r_{2r}$, which means $  \tilde{r}_{r}=r_{2r} $ since $ r_{t1}=1~\mathrm{m}  $. Basically, this means that the IRS in the single-IRS case is located at the position of IRS $ 2 $  in the double-IRS case. 	The size of each IRS element is given by $ d_{\mathrm{H}}\!=\!d_{\mathrm{V}}\!=\!\lambda/8 $ \cite{Bjoernson2020}. The spatial correlation matrix for each IRS is given by \eqref{eq:Element}. Moreover, the path-loss exponents for the Tx-IRS $ 1 $
	and IRS $ 2 $-Rx	links are set as $ 2.2 $, assuming that the IRSs are close to the corresponding Tx, Rx with no special obstacles in the intermediate space while the path-loss exponent for the IRS $ 1 $-IRS $ 2 $ link is set as $ 3 $ since the second IRS is expected to be placed at  a long distance, where the Rx will suffer from stronger signal attenuation. The path-loss exponents of the Tx-IRS $ 1 $
	and IRS $ 2 $-Rx links are set to this value since we have assumed that these links approach LoS conditions due to the smart placement of the IRSs near the Tx and Rx toward better communication. A similar exponent is assumed for the single reflection links, i.e., $ 2.2 $ and $ 3 $ for the first and second link, respectively. The  carrier frequency is $3$~GHz, the system bandwidth is $10$~MHz, $ P=43~\mathrm{dBm} $, and the noise variance is $-94$~dBm.

	Fig. \ref{Fig1}.a illustrates the coverage probability versus the target rate for varying $ N_{1} $ based on  Proposition \ref{proposition:Coverage}. Specifically, given a specific number of total elements $ N=200 $, we vary $ N_{1} $. We observe that by increasing the number of elements in IRS $ 1 $, $ P_{\mathrm{c}} $ increases until $ N_{1}=100 $ elements and then starts decreasing. In other words, we notice 		the highest coverage when the two IRS are implemented with almost the same number of elements, which agrees with \cite{Han2020,Zheng2021}.   Moreover, MC simulations almost coincide with the analytical results, which corroborates that  Alzer's inequality, the selected value of $ M $, and the DE analysis are valid approximations. In the same figure, we show the coverage probability of the single-IRS  SISO channel for $ N=100 $ and $ N=200 $ elements under the same conditions and we observe that the double-IRS cooperative system presents better performance.
	
	In Fig.  \ref{Fig1}.b, we depict the coverage probability versus the target rate for different settings. First,  we focus on the impact of correlation. In particular, in the case of no correlation, $ P_{\mathrm{c}} $ is lower since the covariance matrix of the overall channel, found in the DE SNR,  does not depend on the RBMs, and thus, cannot be optimized. This observation coincides with the results in \cite{Papazafeiropoulos2021a} for single-IRS-assisted communication. In the case of correlation, if the correlation increases, e.g., due to decrease of the inter-element distance, $ P_{\mathrm{c}} $ decreases. Moreover, if random phases shifts are assumed in both IRSs, the coverage is much lower, which means that the RBM optimization definitely improves the performance. Also, we depict the performance of the single-IRS-assisted scenario for the above considerations and we observe that the double-IRS architecture provides better coverage as expected. 
	
In  \ref{Fig1}.c, we demonstrate the convergence of the proposed algorithm, i.e., $ (\mathcal{P}1) $ in \eqref{Maximization}. Specifically, we have 	depicted the coverage probability versus the number of iterations for various numbers of  IRS elements. As can be seen,  the algorithm converges fast in all cases. For example, when  $ N_{1} = 50 $ , the algorithm 	converges in $ 9 $ iterations. In addition, we observe that by increasing the IRS   in terms of its elements, more iterations are required for convergence because the amount of optimization variables	increases and the corresponding search space is enlarged. Also, an increase in terms of  IRS 	elements results in higher complexity of each iteration of the  algorithm as mentioned in Sec. III.B.

%

	\section{Conclusion} \label{Conclusion} 
	In this paper, we derived the coverage probability of a double-IRS-assisted SISO system under the realistic conditions of correlated Rayleigh fading. Specifically, for given RBMs, we obtained its novel expression in closed-form in terms of only large-scale statistics. Moreover, the proposed optimization over the RBMs exhibits a great advantage concerning the reduction of the computational complexity since it can be performed not at each coherence interval but once per several intervals. Among others, numerical results showed the outperformance of double-IRS systems over single-IRS systems and the impact of correlated Rayleigh fading. Future works on the coverage of double-IRS assisted systems could elaborate on the design of multi-user and multi-antenna transmission, and possibly, the impact of Rician fading.
	\begin{appendices}
		\section{Proof of Proposition~\ref{proposition:SNR}}\label{proposition1}
		The proof starts by dividing \eqref{general} with $ \frac{1}{N_{1}N_{2}} $. Hence, we have
	\begin{align}
			\!	&\!\!	\frac{1}{N_{1}N_{2}}{\gamma}\!=\!\frac{\gamma_{0}}{N_{1}N_{2}}\big(\big|\bg_{1}^{\H}\bPhi_{1}\bD\bPhi_{2}\bg_{2}\big|^{2}
		\!\!	+\!|\bg_{1}^{\H}\bPhi_{1}\bu_{1}|^{2}	\!\!+\!|\bu_{2}^{\H}\bPhi_{2}\bg_{2}|^{2}\nn\\
			&		+\!2\mathrm{Re}\!\left(\!\left(\bg_{1}^{\H}\bPhi_{1}\bu_{1}\right)^{*}\!\bg_{1}^{\H}\bPhi_{1}\bD\bPhi_{1}\bg_{2}\right)\!+\!2\mathrm{Re}\!\left(\!\left(\bg_{1}^{\H}\bPhi_{1}\bu_{1}\right)^{*}\!\bu_{2}^{\H}\bPhi_{2}\bg_{2}\right)
			\nn\\
			&
		+\!2\mathrm{Re}\!\left(\!\left(\bu_{2}^{\H}\bPhi_{2}\bg_{2}\right)^{*}\!\bg_{1}^{\H}\bPhi_{1}\bD\bPhi_{1}\bg_{2}\right)
			\!\!\!\!	\big)
			\label{DE_SNR10}\\
					&\asymp \frac{\gamma_{0}}{N_{1}N_{2}}\big(\big|\bg_{1}^{\H}\bPhi_{1}\bD\bPhi_{2}\bg_{2}\big|^{2}
			+\tr\left(\bR_{t1}\bPhi_{1}\bR_{1r}\bPhi_{1}^{\H}\right)\nn\\
			&+\tr\left(\bR_{1r}\bPhi_{1}\bR_{r2}\bPhi_{1}^{\H}\right)\!\!		\big)
			,\label{DE_SNR3}
		\end{align}\noindent		where, in \eqref{DE_SNR3}, we have used \cite[Lem. 4]{Papazafeiropoulos2015a}. Especially,  the last three terms  in \eqref{DE_SNR10} vanish as $ N\to \infty $ because of the independence between different channel vectors. The second and third terms in  \eqref{DE_SNR3} are computed by  applying \cite[Lem. 4]{Papazafeiropoulos2015a} twice. Specifically, first, we condition on $ \bu_{2} $ and $ \bg_{2} $ and apply \cite[Lem. 4]{Papazafeiropoulos2015a} as $ N_{1} \to \infty $, and then, we apply it again over $ \bu_{2} $ and $ \bg_{2} $ as $ N_{2}\to \infty $. Regarding the first term, we have
		\begin{align}
			&\!\frac{1}{N_{1}N_{2}}	\big|\bg_{1}^{\H}\bPhi_{1}\bD\bPhi_{2}\bg_{2}\big|^{2}
			\!\!=\!\frac{1}{N_{1}N_{2}}\!\tr\! \left(\bR_{t1}\bPhi_{1}\bD\bPhi_{2}\bg_{2}\bg_{2}^{\H}\bPhi_{2}^{\H}\bD^{\H}\bPhi_{1}^{\H}\right) \label{term2}\\
			&=\frac{1}{N_{1}N_{2}}\tr \big(\bR_{r2}\bPhi_{2}^{\H}\bD^{\H}\tilde{\bR}_{t1}\bD\bPhi_{2}\big) \label{term3}\\
			&=\!\frac{1}{N_{1}N_{2}}\tr\! \big(\bR_{r2}\bPhi_{2}^{\H}\bR_{1}^{1/2}\tilde{\bD}^{\H}\bR_{2}^{1/2}\tilde{\bR}_{t1}\bR_{2}^{1/2}\tilde{\bD}\bR_{1}^{1/2}\bPhi_{2}\!\big) 		\label{term4}\\
			&=\frac{1}{N_{1}N_{2}}\tr\big(\bR_{2}\tilde{\bR}_{t1} 	\big)\tr \left(\bR_{r2}\bPhi_{2}^{\H}\bR_{1}\bPhi_{2}\right),\label{term5}
		\end{align}
		where, in \eqref{term2},  we have applied \cite[Lem. 4]{Papazafeiropoulos2015a} as $ N_{1}\to \infty $ after conditioning on $ \bg_{2} $ and  $ \bD $. In \eqref{term3}, we have set $ \tilde{\bR}_{t1}=\bPhi_{1}^{\H}\bR_{t1}\bPhi_{1} $ and have applied  \cite[Lem. 4]{Papazafeiropoulos2015a} with respect to $ \bg_{2} $ as $ N_{2}\to \infty $ conditioned on $ \bD $. In \eqref{term4}, we have used that $ \bD=\beta_{12} \bR_{2}^{1/2}\tilde{\bD}\bR_{1}^{1/2}$ with the elements of $ \tilde{\bD}$ being of zero mean and unit variance.  The last step includes application of \cite[Lem. 4]{Papazafeiropoulos2015a} with respect to $ \tilde{\bD} $, which gives $\frac{1}{N_{2}}\tilde{\bD}^{\H}\bR_{2}^{1/2}\tilde{\bR}_{t1}\bR_{2}^{1/2}\tilde{\bD}=\frac{1}{N_{2}}\tr\left(\bR_{2}\tilde{\bR}_{t1} 	\right)\Id_{N_{1}} $. Substitution of \eqref{term5} into \eqref{DE_SNR3} results in the DE SNR $ \bar{\gamma} $ in \eqref{DE_SNR} after some simple algrebraic manipulations.
		
		\section{Proof of Proposition~\ref{proposition:Coverage}}\label{proposition2}
		According to the definition of $ 	P_{\mathrm{c}} $, we have
		\begin{align}
			P_{\mathrm{c}}
			& \!\approx 
			{\mathbb{P}}\bigg( \tilde{g}\!>\frac{T}{\bar{\gamma}}\bigg)\label{coverage6}\\
			&\!\!\!\approx 1-\!\big(\!1-\mathrm{e}^{ -\eta \frac{T}{\bar{\gamma}}}\! \big)^{\!M} ,\label{coverage71}
		\end{align} 
		where after substituting the DE SNR from~\eqref{DE_SNR} into the expression of the coverage probability,  we obtain the right member of~\eqref{coverage6}. In~\eqref{coverage6}, we have also approximated the constant number $1$ by means of the dummy gamma variable $\tilde{g}$, having unit mean and shape parameter $M$  to approximate the constant number one. It has to be mentioned that this approximation becomes tighter as $M$ goes to infinity~\cite{Alzer1997}, since $\lim_{y \to \infty}\frac{y^{y}x^{y-1}\mathrm{e}^{-yx}}{\Gamma\left( y \right)}=\delta\left( x-1 \right)$ with $\delta\left( x \right)$ being Dirac's delta function. Next, in~\eqref{coverage71}, we have applied Alzer's inequality (see~\cite{Bai2015}), where $\eta=M \left(M! \right)^{-\frac{1}{M}}$. As a last step, \eqref{general1} is obtained by making use of the Binomial theorem.
		
		\section{Proof of Lemma~\ref{deriv1}}\label{ArbitraryPDFProof2}
		Application of the  chain rule gives
		\begin{align}
			\pdv{	{P_{\mathrm{c}}}}{\bs_{1,l}^{*}}&=\pdv{	{P_{\mathrm{c}}}}{\bar{\gamma}}\pdv{	\bar{\gamma}}{\bs_{1,l}^{*}},\label{deriv2}
		\end{align}
		where $ \bar{\gamma}$ is given by \eqref{DE_SNR}.  The first derivative in \eqref{deriv2} is obtained as
		\begin{align}
			\pdv{	{P_{\mathrm{c}}}}{\bar{\gamma}}&=\sum^{M}_{n=1} \!\binom{M}{n}\! \frac{\left( -1 \right)^{n+1} n \eta T}{\bar{\gamma}^{2}} \mathrm{e}^{ -n \eta \frac{T}{\bar{\gamma}}}.\label{deriv3}
		\end{align}
		The derivative of $ 	\bar{\gamma}$ with respect to $ \bs_{1,l}^{*} $ is given by
	\begin{align}
			&	\pdv{	\bar{\gamma}}{\bs_{1,l}^{*}}= \!\frac{\gamma_{0}}{N_{1}N_{2}}\big(\!\!\tr \!\left(\bR_{r2}\bPhi_{2}^{\H}\bR_{1}\bPhi_{2}\right)\pdv{\big(\!\left(\diag\!\left(\bR_{t1}\bPhi_{1}\bR_{2} \right)\!\right)^{\T}\bs_{1,l}^{*}\big)}{\bs_{1,l}^{*}}\nn\\
			&+\pdv{\big(\left(\diag\left(\bR_{t1}\bPhi_{1}\bR_{1r} \right)\right)^{\T}\bs_{1,l}^{*}\big)}{\bs_{1,l}^{*}}\big)\label{deriv5}\\
			&=  \!\frac{\gamma_{0}}{N_{1}N_{2}}\big(\!\tr \!\left(\bR_{r2}\bPhi_{2}^{\H}\bR_{1}\bPhi_{2}\right)\diag\left(\bR_{t1}\bPhi_{1}\bR_{2} \right)\nn\\
			&+\diag\left(\bR_{t1}\bPhi_{1}\bR_{1r} \right)\!\!\big),\label{deriv4}
		\end{align}		\noindent where we have applied the property $ \tr\!\big(\!\bA \diag(\bs_{1,l}^{*})\big)=\left(\diag(A)\right)^{\T}\bs_{1,l}^{*} $ in \eqref{deriv5}.  Substitution of \eqref{deriv3}   and \eqref{deriv4} into \eqref{deriv2} provides the desired result.
		
	\end{appendices}
	
	\bibliographystyle{IEEEtran}
	
	\bibliography{mybib}
\end{document}